\documentclass[conference]{IEEEtran}
\usepackage{ragged2e}
\usepackage{floatrow}
\usepackage{amsthm}
\usepackage{bm}
\usepackage[utf8]{inputenc}
\usepackage{soul}            
\usepackage[english]{babel}
\theoremstyle{remark}
\newtheorem{theorem}{Theorem}
\newtheorem{lemma}{Lemma}

\newtheorem{remark}{Remark}

\usepackage{multicol}
\usepackage{epsfig}
\usepackage{epstopdf}
\usepackage{float}
\usepackage{perpage}
\MakeSorted{figure}
\MakeSorted{table}
\usepackage{array}
\usepackage{comment}
\usepackage{graphicx}
\usepackage{amsfonts}
\usepackage{amssymb}
\let\labelindent\relax
\usepackage{enumitem}
\usepackage[english]{babel}
\usepackage[utf8]{inputenc}
\usepackage[linesnumbered,ruled,vlined]{algorithm2e}

\usepackage{cite}
\DeclareMathAlphabet\bmcal{OMS}{cmsy}{b}{n}
\ifCLASSINFOpdf
\else
\fi
\usepackage[cmex10]{amsmath}
\usepackage{array}
\usepackage{fixltx2e}

\usepackage[linesnumbered,ruled,vlined]{algorithm2e}
\usepackage[font=footnotesize,labelfont=bf, figurename=Fig.]{caption} 
\usepackage{subcaption}
\usepackage{soul}
\usepackage{fix-cm}

\newcommand{\colr}{\textcolor{red}}
\newcommand{\colb}{\textcolor{blue}}
\newcommand{\norm}[1]{\left\lVert#1\right\rVert}
\DeclareMathOperator*{\argmax}{arg\,max}

\def\BibTeX{{\rm B\kern-.05em{\sc i\kern-.025em b}\kern-.08em
    T\kern-.1667em\lower.7ex\hbox{E}\kern-.125emX}}
\begin{document}
\title{FD Cell-Free mMIMO: Analysis and Optimization}
\vspace{-0.8in}
\author{
    \IEEEauthorblockN{Soumyadeep Datta\IEEEauthorrefmark{1}\IEEEauthorrefmark{2}, Ekant Sharma\IEEEauthorrefmark{1}, Dheeraj N. Amudala\IEEEauthorrefmark{1}, Rohit Budhiraja\IEEEauthorrefmark{1} and Shivendra S. Panwar\IEEEauthorrefmark{2}}
    \IEEEauthorblockA{\IEEEauthorrefmark{1}Indian Institute of Technology Kanpur \IEEEauthorrefmark{2}NYU Tandon School of Engineering
        \\\{sdatta, ekant, dheeraja, rohitbr\}@iitk.ac.in, \{sd3927, sp1832\}@nyu.edu
        }
}

\maketitle
\begin{abstract}
We consider a  full-duplex cell-free massive multiple-input-multiple-output system with \textit{limited capacity fronthaul links}. We derive its downlink/uplink closed-form spectral efficiency (SE) lower bounds with maximum-ratio transmission/maximum-ratio combining and optimal uniform quantization. To {reduce} carbon footprint, this paper maximizes the non-convex weighted sum energy efficiency (WSEE) via  downlink and uplink power control, and successive convex approximation framework. {We  show that with low fronthaul capacity, the system requires a higher number of fronthaul quantization bits to achieve high SE and WSEE. For high fronthaul capacity, higher number of bits, however, achieves high SE but a reduced WSEE.}
\end{abstract}


\IEEEpeerreviewmaketitle
\section{Introduction} 



Cell-free (CF) massive multiple-input-multiple-output (mMIMO) technology envisions communication without cell boundaries, where all access points (APs), connected via fronthaul links to a central processing unit (CPU), communicate with all the user equipments (UEs). It promises substantial gains in {spectral efficiency (SE)} and fairness over conventional small-cell deployments~\cite{CFvsSmallCells,CellFreeUC}. Full-duplex (FD) 
CF mMIMO is a relatively recent area of interest~\cite{FDCellFree,NAFDCellFree,FDCellFree2}, where APs simultaneously serve downlink and uplink UEs on the same spectral resource.
The existing FD CF mMIMO literature assumes  high-capacity fronthaul links~\cite{FDCellFree,NAFDCellFree,FDCellFree2}. These links, however, have limited capacity, and the information needs to be quantized and sent over them. This aspect has only been investigated for HD CF mMIMO~\cite{CellFreeMaxMinUQ,CellFreeEEUQ}. 

{Energy efficiency (EE)} metric is now being commonly used to design  modern wireless systems \cite{fractionalprogrammingbook}. 
The global EE (GEE) metric, defined as the ratio of the network {throughput} and its total energy consumption, has been used to design CF mMIMO systems~\cite{CellFreeEE,CellFreeEEUQ,FDCellFree2}. A UE with limited energy availability will accord a much higher importance to its EE than an another UE with a sufficient energy supply.  The network-centric GEE metric cannot accommodate such heterogeneous EE requirements~\cite{fractionalprogrammingbook}. 
The weighted sum energy efficiency (WSEE) metric, defined as the weighted sum of individual EEs~\cite{fractionalprogrammingbook}, {is better suited~\cite{fractionalprogrammingbook,WSEEEfrem}}. The current work, unlike the existing CF mMIMO literature~\cite{CellFreeEE,CellFreeEEUQ,FDCellFree2} optimizes WSEE. 

We next list our \textbf{main} contributions in this context:\newline
1) We consider a FD CF mMIMO system with maximal ratio combining/maximal ratio transmission (MRC)/(MRT) at the APs,  and a limited fronthaul with optimal uniform quantization. This is unlike the existing works on FD CF mMIMO~\cite{FDCellFree,NAFDCellFree,FDCellFree2}, which consider perfect high-capacity fronthaul. We derive achievable SE expressions for {uplink and downlink UEs considering multiple antennas at each AP.} 
\newline
{2) We maximize the non-convex WSEE using  successive convex approximation (SCA)-based iterative approach.  We first equivalently recast the  WSEE objective and then locally approximate it as a generalized convex program (GCP).} \newline 
\vspace{-0.4cm}
\section{System model}\label{model}
We consider, as shown in Fig. \ref{fig:0}, a FD CF mMIMO system  where $M$ FD APs serve $K$ single-antenna HD UEs on the same spectral resource, comprising $K_u$ UEs on the uplink and $K_d$ UEs on the downlink, with $K = (K_u + K_d)$. Each AP has $N_t$ transmit antennas and $N_r$ receive antennas, and is connected to the CPU using a limited-capacity fronthaul link, 
which carries quantized uplink/downlink information to/from the CPU. We see from Fig. \ref{fig:0} that due to the FD model, \newline 
     $\bullet$ uplink receive signal of each AP is interfered by its own downlink transmit signal  (intra-AP) and other APs (inter-AP).  \newline 
     $\bullet$ downlink UEs receive transmit signals from uplink UEs, causing uplink downlink interference (UDI). 
     
  Additionally, the UEs experience multi-UE interference (MUI) as the APs serve them on the same spectral resource. 
  \vspace{-0.8cm}
\begin{figure}[htbp]
	\centering
	\includegraphics[width = 0.8\textwidth]{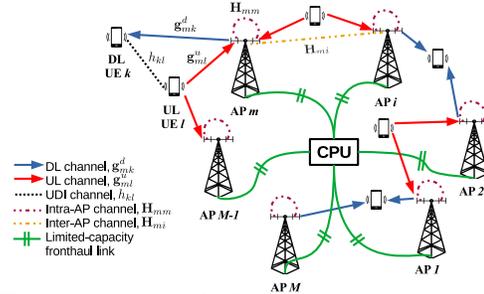}
	\vspace{-0.8cm}
	\caption{System Model for FD CF mMIMO communications\vspace{-0.0cm}}
	\label{fig:0}
\end{figure}
%
\textbf{Channel description:} 
The vector $\bm{g}^{d}_{mk} \in \mathbb{C}^{N_t \times 1}$ is the channel from the $k$th downlink UE to the transmit antennas of the $m$th AP, and $\bm{g}^{u}_{ml} \in \mathbb{C}^{N_r \times 1}$ is the channel from the $l$th uplink UE to the receive antennas of the $m$th AP\footnote{We, henceforth, consider $k = 1 \text{ to } K_d, l = 1 \text{ to } K_u$ and $m = 1 \text{ to } M$, to avoid repetition, unless mentioned otherwise.}. We model them as $\bm{g}^{d}_{mk} = (\beta^{d}_{mk})^{1/2} \Tilde{\bm{g}}^{d}_{mk}$ and  $\bm{g}^{u}_{ml} = (\beta^{u}_{ml})^{1/2} \Tilde{\bm{g}}^{u}_{ml}$. Here $\beta^{d}_{mk}$ and $ \beta^{u}_{ml} \in \mathbb{R}$ are corresponding large scale fading coefficients, which are same for all antennas at the $m$th AP~\cite{CFvsSmallCells, FDCellFree}. The vectors $\Tilde{\bm{g}}^{d}_{mk}$ and $\Tilde{\bm{g}}^{u}_{ml}$ denote small scale fading  with independent and identically distributed (i.i.d.) $\mathcal{CN}(0,1)$ entries. The UDI channel between the $k$th downlink UE and $l$th uplink UE is modeled similar to~\cite{FDCellFree, NAFDCellFree}, as
$h_{kl} = (\Tilde{\beta}_{kl})^{1/2} \Tilde{h}_{kl}$, 
where $\Tilde{\beta}_{kl}$ is the large scale fading coefficient and $\Tilde{h}_{kl} \sim \mathcal{CN}(0,1)$ is the small scale fading. The  inter- and intra-AP channels from the transmit antennas of the $i$th AP to the receive antennas of the $m$th AP are denoted as $\bm{H}_{mi} \in \mathbb{C}^{N_r \times N_t},\, i=1 \text{ to } M$.
\textbf{Uplink channel estimation:} \label{ul_ch_est}
We assume a coherence interval of duration $T_c$ (in s) with $\tau_c$ samples, which  is divided into: a) channel estimation phase of $\tau_t$ samples, and b) downlink and uplink data transmission of ($\tau_c$ - $\tau_t$) samples. We consider $\tau^d_t$ and $\tau^u_t$ pilots for the downlink and uplink UEs, respectively, with $\tau_t = \tau^d_t + \tau^u_t$.  All the downlink (resp. uplink) UEs, on the same spectral resource, simultaneously transmit $\tau^d_t$ (resp. $\tau^u_t$)-length uplink pilots to the APs, which estimate the channel to perform MRT (resp. MRC). 
The $k$th downlink UE (resp. $l$th uplink UE) transmits pilot signals $\sqrt{\tau^d_t}\bm{\varphi}^{d}_{k} \in \mathbb{C}^{\tau^d_t \times 1}$ (resp. $\sqrt{\tau^u_t}\bm{\varphi}^{u}_{l} \in \mathbb{C}^{\tau^u_t \times 1}$). We assume that the pilots  i) have unit norm i.e., $\norm{\bm{\varphi}^{u}_{l}} = \norm{\bm{\varphi}^{d}_{k}} = 1$; {ii) have equal normalized transmit signal-to-noise ratio (SNR), $\rho_t$}; and iii) are intra-set orthonormal i.e. 
$(\bm{\varphi}^{u}_{l})^{H} \bm{\varphi}^{u}_{l'} = 0\, \forall l \neq l' \text { and  } (\bm{\varphi}^{d}_{k})^{H} \bm{\varphi}^{d}_{k'} = 0\, \forall  k \neq k'$~\cite{CellFreeEE,FDCellFree}. 
Therefore, $\tau^d_t \geq K_d$ and $\tau^u_t \geq K_u$~\cite{FDCellFree}. 

The pilots transmitted by the downlink (resp. uplink) UEs are received by the transmit (resp. receive) antennas of the APs. The transmit antennas operate in the receive mode while estimating channel~\cite{FDCellFree}. The APs then project the received signal on the pilot signals to compute the linear minimum-mean-squared-error (MMSE) channel estimates~\cite{FDCellFree}, denoted as $\hat{\bm{g}}^{d}_{mk}$ and 
$\hat{\bm{g}}^{u}_{ml}$, respectively. The estimation error vectors are defined as $\bm{e}^{u}_{ml} \!\triangleq\! \bm{g}^{u}_{ml} \!\!-\!\! \hat{\bm{g}}^{u}_{ml}$ and  $\bm{e}^{d}_{mk} \!\triangleq\! \bm{g}^{d}_{mk} \!\!-\!\! \hat{\bm{g}}^{d}_{mk}$. With MMSE  estimation, $\hat{\bm{g}}^{d}_{mk}, \bm{e}^{d}_{mk}$ and $\hat{\bm{g}}^{u}_{ml}, \bm{e}^{u}_{ml}$ are mutually independent and their individual terms are i.i.d. with pdf $\mathcal{CN}(0,\gamma^{d}_{mk}),\mathcal{CN}(0,\beta^{d}_{mk}\!\!-\!\!\gamma^{d}_{mk}), \mathcal{CN}(0,\gamma^{u}_{ml}), \mathcal{CN}(0, \beta^{u}_{ml} \!\!-\!\! \gamma^{u}_{ml})$ 
respectively, with $\gamma^{d}_{mk} \!\!=\!\! \frac{\tau^d_t \rho_t (\beta^{d}_{mk})^2}{\tau^d_t \rho_t \beta^{d}_{mk} \!+\! 1}$ and $\gamma^{u}_{ml} \!\!=\!\! \frac{\tau^u_t \rho_t (\beta^{u}_{ml})^2}{\tau^u_t \rho_t \beta^{u}_{ml} \!+\! 1}$~\cite{FDCellFree}. 
\textbf{Downlink data transmission:}
The CPU chooses  message symbol with pdf  $\mathcal{CN}(0,1)$ for the $k$th downlink UE, which is denoted as $s^d_k$. {To send this symbol to the $m$th AP via the limited-capacity fronthaul, the CPU first multiplies $s^d_k$ with a power-control coefficient $\eta_{mk}$, and then quantizes it.}  The $m$th AP uses the MMSE channel estimates to perform MRT on the quantized signal, and generates its transmit signal as following
\begin{align} \label{dltransmitsignal}
   \bm{x}^{d}_{m} 
   &= \sqrt{\rho_d} \sum_{k \in \kappa_{dm}}  (\hat{\bm{g}}^{d}_{mk})^{*} (\Tilde{a}\sqrt{\eta_{mk}} s^{d}_{k} + \varsigma^{d}_{mk}).
\end{align}
Here $\rho_d$ is the normalized maximum transmit SNR at each AP. {The $m$th AP, due to limited fronthaul capacity, serves only the UEs in the subset $\kappa_{dm} \subset \{1, \dots, K_d\}$, as discussed later in Section~\ref{apsellimfh}}. The CPU sends downlink symbols for UEs in the set $\kappa_{dm}$. 
{The quantization operation is modeled} as a multiplicative attenuation, $\Tilde{a}$, and an additive distortion, $\varsigma^{d}_{mk}$, for the $k$th downlink UE in the fronthaul link between the CPU and the $m$th AP~\cite{CellFreeMaxMinUQ,CellFreeEEUQ}. {The constants $\Tilde{a}$ and $\Tilde{b}$ depend on the number of fronthaul quantization bits, $\nu_m$, and $\mathbb{E}\{\left(\varsigma^{d}_{mk}\right)^{2}\} = (\Tilde{b}-\Tilde{a}^2) \mathbb{E}\left\{|\sqrt{\eta_{mk}} s^d_k|^2\right\} =  (\Tilde{b}-\Tilde{a}^2)\eta_{mk}$~\cite{CellFreeMaxMinUQ}.} 
Satisfying the average transmit SNR constraint, $\mathbb{E}\{\|\bm{x}^{d}_{m}\|^{2}\} \leq \rho_d$, we get
\begin{equation}
\rho_d \Tilde{b} \sum_{k \in \kappa_{dm}} \!\!\!\! \eta_{mk} \mathbb{E}\{\|\hat{\bm{g}}^d_{mk}\|^{2}\} \leq \rho_d \Rightarrow \Tilde{b} \sum_{k \in \kappa_{dm}} \!\!\!\! \gamma^{d}_{mk} \eta_{mk} \leq \frac{1}{N_t}.\!\! \label{cons1} 
\end{equation}
\textbf{Uplink data transmission:} The $K_u$ uplink UEs also transmit simultaneously to all the $M$ APs on the same spectral resource. The $l$th uplink UE transmits its signal $x^{u}_{l} = \sqrt{\rho_u \theta_l} s^{u}_{l}$ with  
{$s^u_l \sim \mathcal{CN}(0,1)$} being its  message symbol, $\rho_u$  being the  maximum uplink transmit SNR and $\theta_l$ being the power control coefficient. To {satisfy} the average {transmit} SNR constraint 
\begin{equation}
\mathbb{E}\{|x^{u}_{l}|^{2}\} \leq \rho_u \Leftrightarrow 0 \leq \theta_l \leq 1. \label{cons2} 
\end{equation}

The FD APs not only receive uplink UE signals but also their own  downlink transmit signals  and that of other APs,  referred to as intra-AP and inter-AP interference. The intra and inter-AP interference channels vary extremely slowly and the FD APs can thus estimate them with very low pilot overhead~\cite{NAFDCellFree}. We assume that these channel estimates are imperfect. The receive antenna array of each AP, similar to~\cite{FDCellFree, NAFDCellFree}, can therefore only partially mitigate the intra- and inter-AP interference. The residual intra-/inter-AP interference (RI) channel $\bm{H}_{mi} \in \mathbb{C}^{N_r \times N_t}, \, i = 1 \text{ to } M$, similar to \cite{FDCellFree, NAFDCellFree},  is modeled as a Rayleigh faded channel with entries i.i.d.$\sim \mathcal{CN}(0, \gamma_{\text{RI},mi})$. Here $\gamma_{\text{RI},mi} \triangleq \beta_{\text{RI},mi} \gamma_{\text{RI}}$, with $\beta_{\text{RI},mi}$ being the large scale fading coefficient from the $i$th AP to the $m$th AP, and $\gamma_{\text{RI}}$ being the RI power after interference cancellation. Using \eqref{dltransmitsignal}, received uplink signal {at the $m$th AP:} 
\begin{align} 
    \notag \bm{y}^{u}_{m} \!&=\! \sum_{l=1}^{K_u} \!\bm{g}^{u}_{ml} x^{u}_{l} \! + \! \sum_{i=1}^{M} \!\bm{H}_{mi} \bm{x}^{d}_{i} \!+\! \bm{w}^{u}_{m} \!= \! \sqrt{\rho_u} \! \sum_{l=1}^{K_u} \!  \bm{g}^{u}_{ml} \sqrt{\theta_l} s^{u}_{l} \\
    &\quad+ \! \sqrt{\rho_d} \sum_{i=1}^{M} \sum_{k \in \kappa_{di}} \!\!\! \bm{H}_{mi} (\hat{\bm{g}}^{d}_{ik})^{*} (\Tilde{a} \sqrt{\eta_{ik}} s^{d}_{k} + \varsigma^{d}_{ik}) \! + \! \bm{w}^{u}_{m}.\!\! \label{rxulsignal}
\end{align}
Here $\bm{w}^{u}_{m} \in \mathbb{C}^{N_r \times 1}$ is the additive receiver noise at the $m$th AP with i.i.d. entries $\sim \mathcal{CN}(0,1)$. \newline
\textbf{Decoding of received signal on the downlink and uplink:} The $k$th downlink UE receives its desired message signal from a subset of all APs, denoted as $\mathcal{M}^{d}_{k} \subset \{1, \dots, M\}$. 
The $m$th AP serves the $k$th downlink UE iff $k \in \kappa_{dm} \Leftrightarrow m \in \mathcal{M}^d_k$. 

The $m$th AP performs MRC on the received signal, $\bm{y}^{u}_{m}$, for the $l$th uplink UE using $(\hat{\bm{g}}^{u}_{ml})^{H}$. It quantizes the combined signal before sending it to CPU
, modeled using a constant attenuation $\Tilde{a}$, and additive distortion, $\varsigma^{u}_{ml}$, for the $l$th uplink UE in the fronthaul link between the $m$th AP and the CPU, with $\mathbb{E}\{\left(\varsigma^{u}_{ml}\right)^{2}\} \!\!=\!\! (\Tilde{b} \!\!-\!\! \Tilde{a}^{2})\mathbb{E}\left\{\big|(\bm{g}^{u}_{ml})^{H}\bm{y}_m\big|^{2}\right\}$~\cite{CellFreeMaxMinUQ}. The CPU receives signals for the $l$th uplink UE 
{from the APs in the subset $\mathcal{M}^{u}_{l} \subset \{1, \dots, M\}$,} 
due to limited fronthaul capacity, as discussed in detail later in Section \ref{apsellimfh}. We denote the subset of uplink UEs served by the $m$th AP as $\kappa_{um} \subset \{1, \dots, K_u\}$. The $m$th AP serves the $l$th uplink UE iff $l \in \kappa_{um} \Leftrightarrow m \in \mathcal{M}^u_l$.

Using~\eqref{dltransmitsignal}-\eqref{rxulsignal}, the signal received by the $k$th downlink UE and by CPU for the $l$th uplink UE are~\eqref{dlsignal} and~\eqref{ulsignal}, respectively.
\begin{figure*}
\begin{align}
    \notag r^{d}_{k}& = \sum_{m=1}^{M} (\bm{g}^{d}_{mk})^{T} \bm{x}^{d}_{m} + \sum_{l=1}^{K_u} h_{kl} x^{u}_{l} + w^{d}_{k} = \underbrace{\Tilde{a}\sqrt{\rho_d} \sum_{m \in \mathcal{M}^{d}_{k}} \sqrt{\eta_{mk}}(\bm{g}^{d}_{mk})^{T}  (\hat{\bm{g}}^{d}_{mk})^{*} s^{d}_{k}}_{\text{message signal}} + \underbrace{\Tilde{a}\sqrt{\rho_d} \sum_{m=1}^{M} \sum_{q \in \kappa_{dm} \setminus k} \sqrt{\eta_{mq}}(\bm{g}^{d}_{mk})^{T}  (\hat{\bm{g}}^{d}_{mq})^{*} s^{d}_{q}}_{\text{multi-UE interference, MUI}^{d}_{k}} \\
    & \quad + \underbrace{\sum_{l=1}^{K_u} \sqrt{\rho_u} \sqrt{\theta_l} h_{kl}  s^{u}_{l}}_{\text{uplink downlink interference, UDI}^{d}_{k}} + \underbrace{\sqrt{\rho_d} \sum_{m=1}^{M} \sum_{q \in \kappa_{dm}} (\bm{g}^{d}_{mk})^{T}  (\hat{\bm{g}}^{d}_{mq})^{*} \varsigma^{d}_{mq}}_{\text{total quantization distortion, TQD}^{d}_{k}} + \underbrace{w_{k}^{d}}_{\text{AWGN at receiver}}, \label{dlsignal} \\
    \notag r^{u}_{l} &= \underbrace{\Tilde{a}\sum_{m \in \mathcal{M}^{u}_{l}} \sqrt{\rho_u} \sqrt{\theta_l} (\hat{\bm{g}}^{u}_{ml})^{H} \bm{g}^{u}_{ml} s_{l}^{u}}_{\text{message signal}} + \underbrace{\Tilde{a} \sum_{m \in \mathcal{M}^{u}_{l}} \sum_{q=1, q \neq l}^{K_u} \sqrt{\rho_u} \sqrt{\theta_q} (\hat{\bm{g}}^{u}_{ml})^{H} \bm{g}^{u}_{mq} s_{q}^{u}}_{\text{multi-UE interference, MUI}^{u}_{l}} \\
    &\quad+\! \underbrace{\Tilde{a}  \sum_{m \in \mathcal{M}^{u}_{l}} \sum_{i=1}^{M}  \sqrt{\rho_d} \sum_{k \in \kappa_{di}} (\hat{\bm{g}}^{u}_{ml})^{H} \bm{H}_{mi} (\hat{\bm{g}}^{d}_{ik})^{*} (\Tilde{a} \sqrt{\eta_{ik}} s^{d}_{k} + \varsigma^{d}_{ik})}_{\text{residual interference (intra-AP and inter-AP), RI}^{u}_{l}} +  \underbrace{\Tilde{a}  \sum_{m \in \mathcal{M}^{u}_{l}} (\hat{\bm{g}}^{u}_{ml})^{H} \bm{w}^{u}_{m}}_{\text{AWGN at APs, N}^{u}_{l}} +  \underbrace{\sum_{m \in \mathcal{M}^{u}_{l}} \varsigma^{u}_{ml}}_{\text{total quantization distortion, TQD}^u_l}. \label{ulsignal}
\end{align}
\hrulefill
\vspace{-0.8cm}
\end{figure*}
\textbf{Quantization, limited fronthaul and AP selection:}\label{apsellimfh}
The fronthaul link between the $m$th AP and the CPU uses $\nu_{m}$ bits to quantize the real and imaginary parts of transmit signal of the {$k$th downlink UE, i.e.,  $\sqrt{\eta_{mk}} s^d_k$, and the uplink receive signal after MRC, i.e., $(\hat{\bm{g}}^{u}_{ml})^{H} \bm{y}^u_m$.} 
Due to the limited-capacity fronthaul, the $m$th AP serves only $K_{um} (\triangleq |\kappa_{um}|)$ and $K_{dm} (\triangleq |\kappa_{dm}|)$ UEs on the uplink and downlink, respectively~\cite{CellFreeMaxMinUQ,CellFreeEEUQ}.  For each UE, we recall that there are $(\tau_c - \tau_t)$ data samples in each coherence interval of duration $T_c$. 

The fronthaul data rate between the $m$th AP and the CPU:
\begin{equation} \label{fhrate}
    R_{\text{fh},m} = \frac{2\nu_m(K_{dm} + K_{um}) (\tau_c - \tau_t)}{T_c}. 
\end{equation}
The fronthaul capacity for the $m$th AP is $C_{\text{fh},m}$, implying
\begin{equation} \label{limfronthaul}
R_{\text{fh},m} \leq C_{\text{fh},m} \Rightarrow   \nu_m \cdot (K_{um} + K_{dm}) \leq \frac{C_{\text{fh},m} T_c}{2(\tau_c - \tau_t)}. 
\end{equation}
{To ensure uplink and downlink fairness for the FD system, we set the same maximum limit to the number of uplink as well as downlink UEs an AP can serve, which satisfies \eqref{limfronthaul}.} 
\vspace{-0.2cm}
\begin{lemma}
The maximum number of uplink and downlink UEs served by the $m$th AP when connected via a limited optical fronthaul to the CPU with capacity $C_{\text{fh},m}$ is given as 
\begin{equation} \label{maxUEs}
    \Bar{K}_{um} = \Bar{K}_{dm} =  \left \lfloor{\frac{C_{\text{fh},m} T_c}{4 (\tau_c - \tau_t) \nu_m}}\right \rfloor. 
\end{equation}
\end{lemma}
\vspace{-0.1cm}
\begin{proof}
Let $\Bar{K}_{um}$ and $\Bar{K}_{dm}$ denote the maximum number of uplink and downlink UEs served by the $m$th AP. Consider $\Bar{K}_{um} = \Bar{K}_{dm}$ for uplink downlink fairness. Using \eqref{limfronthaul}, 
\begin{equation}
    \notag \Bar{K}_{um} = \Bar{K}_{dm} \leq \frac{C_{\text{fh},m} T_c}{4(\tau_c - \tau_t) \nu_m}. 
\end{equation}
The lemma follows from definition of floor function $\lfloor{ \cdot }\rfloor$.
\end{proof}
Using the maximum limits obtained in~\eqref{maxUEs}, we assign $K_{um} \!\! = \!\! \min \{K_u,\! \Bar{K}_{um}\}, K_{dm} \!\! = \!\! \min \{K_d,\! \Bar{K}_{dm}\}$. 
We define the procedure for AP selection {for FD CF mMIMO with limited fronthaul constraint~\eqref{maxUEs}}, 
by extending the procedure in~\cite{CellFreeMaxMinUQ} as: \newline 
    $\bullet$ The $m$th AP sorts the uplink and downlink UEs connected to it in descending order based on their channel gains ($\beta^{u}_{ml}$ and $\beta^{d}_{mk}$, respectively) and chooses the first $K_{um}$ uplink UEs and $K_{dm}$ downlink UEs 
    to populate $\kappa_{um}$ and $\kappa_{dm}$, respectively. 
    \newline
    $\bullet$ 
   {We} populate the sets $\mathcal{M}^u_l$ and $\mathcal{M}^{d}_{k}$, respectively, using the axioms $l \in \kappa_{um} \Leftrightarrow m \in \mathcal{M}^u_l$ and $k \in \kappa_{dm} \Leftrightarrow m \in \mathcal{M}^d_k$. \newline
    $\bullet$ If an uplink or downlink UE is found with no serving AP, we follow the procedure in Algorithm~\ref{algo0} to assign it the AP with the best channel conditions, while satisfying~\eqref{limfronthaul}.
    \begin{algorithm}
	\footnotesize
        \DontPrintSemicolon 
        \For{$k \gets 1$ \textbf{to} $K_d$}{
        \lIf{$\mathcal{M}^d_k = \phi$}{\linebreak
            Sort the APs in descending order of channel gains, $\beta^d_{mk}$, and find the AP $n$ with the largest channel gain. \linebreak
            For this $n$th AP, sort downlink UEs in $\kappa_{dn}$ in descending order of channel gains and find the $q$th downlink UE with minimum channel gain and \textit{at least one more connected AP}. \linebreak
            Remove the $q$th downlink UE from the set $\kappa_{dn}$ and add the $k$th downlink UE to it.
        }
        }
        Repeat the same procedure for all the uplink UEs $l = 1$ to $K_u$. 
        \caption{Fair AP selection for disconnected UEs}\label{algo0}
    \end{algorithm}
\section{Achievable spectral efficiency} \label{rateexpresns}
We now derive the achievable SE for the $k$th downlink UE and the $l$th uplink UE, denoted respectively as $S^{d}_{k}$ and $S^{u}_{l}$. We use $\varepsilon \!\triangleq\! \{d,u\}$ to denote downlink and uplink, respectively; $\phi \!\triangleq\! \{k,l\}$ to denote the $k$th downlink UE and the $l$th uplink UE, respectively; and $\upsilon^{\varepsilon}_{m\phi} \!\triangleq\! \{\eta_{mk} \text{ for } \phi\!=\!k, \theta_{l} \!\text{ for }\! \phi\!=\!l\}$. We employ use-and-then-forget (UatF) technique to derive SE lower bounds~\cite{CFvsSmallCells}. We rewrite the received signal at the $k$th downlink UE in~\eqref{dlsignal} and {at the CPU for the $l$th uplink UE in~\eqref{ulsignal}, as follows} 
\begin{small}
\begin{align}
    r^{\varepsilon}_{\phi}\!&=\! \underbrace{\Tilde{a} \!\!\!\! \sum_{m \in \mathcal{M}^{\varepsilon}_{\phi}} \!\!\!\! \sqrt{\rho_{\varepsilon}} \sqrt{\upsilon^{\varepsilon}_{m\phi}} \mathbb{E}\Big\{(\hat{\bm{g}}^{\varepsilon}_{m\phi})^{H} \bm{g}^{\varepsilon}_{m\phi}\Big\}s^{\varepsilon}_{\phi}}_{\text{desired signal, DS}^{\varepsilon}_{\phi}}+ n^{\varepsilon}_{\phi}, \!\!\!\! \label{sigeff1}
\end{align}
\end{small} 
where the first term, denoted as DS$^{\varepsilon}_{\phi}$, is the ``true" desired signal received over the channel mean. {The term $n^{\varepsilon}_{\phi}$ denotes the effective additive noise. It contains, besides the various interferences (MUI, RI, UDI), noise (N) and quantization distortion (TQD) terms as expressed in~\eqref{dlsignal}-\eqref{ulsignal}, an additional beamforming uncertainty term, denoted as BU$^{\varepsilon}_{\phi}$, which is a ``signal leakage" term received over deviation of channel from mean. It is expressed as} 
\begin{small}
\begin{align}
     \!\!\!\!\!\text{BU}^{\varepsilon}_{\phi}\!\!&=\!\Tilde{a} \sqrt{\rho_{\varepsilon}} \!\!\!\!\!\! \sum_{m \in \mathcal{M}^{\varepsilon}_{\phi}} \!\!\!\!\! \sqrt{\upsilon^{\varepsilon}_{m\phi}}((\bm{g}^{\varepsilon}_{m\phi})^{T}\! (\hat{\bm{g}}^{\varepsilon}_{m\phi})^{*} \!\!\! -\!\mathbb{E}\{(\bm{g}^{\varepsilon}_{m\phi})^{T} (\hat{\bm{g}}^{\varepsilon}_{m\phi})^{*}\})s^{\varepsilon}_{\phi}, \!\!\!\!\!\! \label{BU}  
\end{align}
\end{small}
It is easy to see that $n^{\varepsilon}_{\phi}$ is uncorrelated with 
$\text{DS}^{\varepsilon}_{\phi}$. We treat them as worst-case additive Gaussian noise, which is guaranteed to be a tight approximation for massive MIMO~\cite{FDCellFree}. {We define $\tau_f \!\!=\!\! \left(\!\frac{\tau_c - \tau_t}{\tau_c}\!\right)$, $A^d_{mk} \!\!=\!\! \Tilde{a} N_t \sqrt{\rho_d} \gamma^d_{mk}$, $B^d_{kmq} \!\!=\!\! \Tilde{b} N_t \rho_d \beta^{d}_{mk} \gamma^{d}_{mq}$, $D^d_{kl} = \rho_u \Tilde{\beta}_{kl}$, $A^u_l \!\!=\!\! \Tilde{a}^{2} N^2_r \rho_u (\!\!\!\!\sum\limits_{m \in \mathcal{M}^{u}_{l}}\!\! \gamma^{u}_{ml})^{2}\!\!$, $B^u_{lq} \!\!\!=\!\!\! \Tilde{b} N_r \rho_u \!\!\!\! \sum\limits_{m \in \mathcal{M}^{u}_{l}} \!\!\!\! \gamma^{u}_{ml}\beta^{u}_{mq}$, $D^u_{lik} \!\!\!=\!\!\! \Tilde{b}^2 N_r N_t \rho_d \gamma^{d}_{ik}\!\!\!\! \sum\limits_{m \in \mathcal{M}^{u}_{l}}\!\!\!\! \gamma^{u}_{ml} \beta_{\text{RI},mi} \gamma_{\text{RI}}$,$E^u_l\!\!=\!\!(\Tilde{b}\!\!-\!\!\Tilde{a}^2)N^2_r \rho_u \!\!\!\!\!\! \sum\limits_{m \in \mathcal{M}^u_l}\!\!\!\! (\gamma^{u}_{ml})^{2}$, and $F^u_l \!\!=\!\! \Tilde{b} N_r \!\!\sum_{m \in \mathcal{M}^{u}_{l}}\!\! \gamma^{u}_{ml}$}. Using \eqref{sigeff1}-\eqref{BU}, 
we derive an achievable SE lower bound for each UE. 
\begin{theorem}
An achievable lower bound to the SE for the $k$th downlink UE with MRT and $l$th uplink UE with MRC is
\begin{align} 
    \!\!S^{d}_{k} \!&=\! \tau_f \log_2 \Bigg(1 \!+\! \frac{(\sum_{m \in \mathcal{M}^{d}_{k}} \!\! A^d_{mk} \sqrt{\eta_{mk}})^{2}}{\sum\limits_{m=1}^{M} \sum\limits_{q \in \kappa_{dm}}\!\!\!\!\! B^d_{kmq} \eta_{mq} \!+\! \sum\limits_{l=1}^{K_u} \!\! D^d_{kl} \theta_l \!+\! 1}\Bigg), \label{dlrate} \\
    \!\!S^{u}_{l} \!&=\! \tau_f \! \log_2 \! \Bigg(\!1 \!+\! \frac{A^u_l \theta_l}{\sum\limits_{q=1}^{K_u} \!\! B^u_{lq} \theta_q \!\!+\!\! \sum\limits_{i=1}^{M} \! \sum\limits_{k \in \kappa_{di}}\!\!\!\!\! D^u_{lik} \eta_{ik} \!\!+\!\! E^u_l \theta_l \!\!+\!\! F^u_l}\!\Bigg). \label{ulrate} 
\end{align} 
 {Here $\bm{\eta} \!\! \triangleq \!\! \{\eta_{mk}\} \!\!\in\!\! \mathbb{C}^{M \times K_d}$, $\bm{\Theta} \!\!\triangleq\!\! \{ \theta_l\}  \!\!\in\!\! \mathbb{C}^{K_u \times 1}$ and $\bm{\nu} \!\! \triangleq\!\! \{\nu_m\} \!\! \in \!\! \mathbb{C}^{M \times 1}$ are the variables on which the SE for each UE depends.} 
\end{theorem}
\begin{proof}
{Refer to Appendix~\ref{SINRterms}.}
\end{proof}
\vspace{-0.2cm}
\section{WSEE Maximization for FD CF mMIMO} \label{ADMM}
\vspace{-0.2cm}
We now maximize WSEE by calculating the optimal {power control coefficients $\{\bm{\eta}^{*}, \bm{\Theta}^{*}\}$.} 
We denote $\varepsilon\!\triangleq\! \{d,u\}$ for the downlink and uplink, respectively; $\phi \!\triangleq\! \{k,l\}$ for the $k$th downlink UE and $l$th uplink UE, respectively; and define the individual EEs for each UE as $\text{EE}^{\varepsilon}_{\phi} \!\!=\!\! \frac{B \cdot S^{\varepsilon}_{\phi}}{p^{\varepsilon}_{\phi}}$,~\cite{WSEEEfrem}, with $B$ being the bandwidth and $p^{\varepsilon}_{\phi}$ denoting the individual power consumption for each UE. The power consumed by the $k$th downlink UE and the $l$th uplink UE are obtained respectively as~\cite{CellFreeEEUQ}
\begin{align}
p^d_k &= P_{\text{fix}} + N_t \rho_d N_0 \sum_{m \in \mathcal{M}^{d}_{k}} \frac{1}{\alpha_m}\gamma^{d}_{mk} \eta_{mk}  + P^{d}_{\text{tc},k}, \label{dlpower} \\
p^u_l &= P_{\text{fix}} + \rho_u N_0 \frac{1}{\alpha'_l} \theta_l + P^{u}_{\text{tc},l}. \label{ulpower} 
\end{align}
Here $\alpha_m, \alpha'_l$ are power amplifier efficiencies at the $m$th AP and the $l$th uplink UE respectively~\cite{FDCellFree}, $N_0$ is noise power. The $P_{\text{fix}}$ is the {fixed per-UE  power consumed} by the APs i.e., 
\begin{align}
    P_{\text{fix}} \!&= \! \frac{1}{K} \sum_{m=1}^{M} \! \left( P_{0,m} \!+\! (N_t\!+\!N_r) P_{\text{tc},m} \!+\! P_{\text{ft}}\frac{R_{\text{fh},m}}{C_{\text{fh},m}} \right).\!\! \label{fixpower}
\end{align} 
Here 
$P_{\text{tc},m}$, $P^{d}_{\text{tc},k}, P^{u}_{\text{tc},l}$ are transceiver chain power consumption at each antenna of the $m$th AP, $k$th downlink UE and $l$th uplink UE, respectively. The fronthaul power consumed  by the $m$th AP has a fixed component $P_{0,m}$, and a traffic-dependent component, {proportional to the data rate $R_{\text{fh},m}$, attaining} a maximum of $P_{\text{ft}}$ at capacity $C_{\text{fh},m}$. 

The WSEE is defined as weighted sum of individual EEs~\cite{fractionalprogrammingbook} %
\begin{align}
    \text{WSEE} \!\! &= \!\! \sum_{k=1}^{K_d} \! w^d_k \text{EE}^d_k \!+\! \sum_{l=1}^{K_u} \! w^u_l \text{EE}^u_l \!\stackrel{\bigtriangleup}{=}\!B\! \left(\sum_{k=1}^{K_d}\! w^d_k \frac{S^{d}_{k}}{p^d_k} \!+\! \sum_{l=1}^{K_u}\! w^u_l \frac{S^{u}_{l}}{p^u_l}\right)\!\!, \!\!\!\! \nonumber 
\end{align}
where $w^{\varepsilon}_{\phi}$
are weights {assigned to downlink and uplink UEs to account for heterogeneous EE requirements of UEs~\cite{WSEEEfrem}}. 

The WSEE maximization problem is formulated as 
\begin{subequations} 
\begin{alignat}{2}
\notag \!\!\!\! \textbf{P1}: & \underset{\substack{\bm{\eta \text{, } \Theta \text{, } \nu}}}{\mbox{max}} && B \left(\sum_{k=1}^{K_d} w^d_k \frac{S^{d}_{k}(\bm{\eta, \Theta,\nu})}{p^d_k(\bm{\eta,\nu})} + \sum_{l=1}^{K_u} w^u_l \frac{S^{u}_{l}(\bm{\eta, \Theta, \nu})}{p^u_l(\bm{\Theta,\nu})} \right) \\ 
& \;\; \text{ s.t. } \; && S^d_k (\bm{\eta, \Theta,\nu}) \geq S^d_{ok} \text{, } S^u_l (\bm{\eta, \Theta,\nu}) \geq S^u_{ol}, \label{P1const1} \\ 
&&&  R_{\text{fh},m} \leq C_{\text{fh},m}, \eqref{cons1}, {\eqref{cons2}}. \label{P1const3}  
\end{alignat}
\end{subequations}
The quality-of-service (QoS) constraints in \eqref{P1const1} guarantee a minimum SE, 
$S^d_{ok}$ and $S^u_{ol}$, for each downlink and uplink UE respectively. 
We observe that the number of quantization bits, $\bm{\nu}$, {makes \textbf{P1} a} difficult integer optimization problem~\cite{CellFreeEEUQ}. {We therefore fix $\bm{\nu}$ such that it satisfies the first constraint in~\eqref{P1const3}~\cite{CellFreeEEUQ} 
and numerically investigate {$\bm{\nu}$} in Section \ref{simresults}.}
{We now linearize the non-convex objective in {\textbf{P1}, after omitting the constant $B$,} using epigraph transformation~\cite{boyd2004convex}}:
\begin{alignat}{2}
\notag \!\!\!\! {\textbf{P2}}:& \underset{\substack{\bm{\eta\text{, }\Theta\text{, }f^{d}\text{, }f^{u}}}}{\mbox{max}} && \sum_{k=1}^{K_d} w^d_k f^d_k + \sum_{l=1}^{K_u} w^u_l f^u_l  \\
& \;\; \text{ s.t. } \; &&  f^d_k \leq \frac{S^{d}_{k}(\bm{\eta, \Theta})}{p^d_k(\bm{\eta})} \text{, } f^u_l \leq \frac{S^{u}_{l}(\bm{\eta, \Theta})}{p^u_l(\bm{\Theta})}, \label{P3const1} \\ 
&&& \eqref{cons1}, \eqref{cons2}, {\eqref{P1const1}}. \nonumber 
\end{alignat}
Here, $\bm{f}^{\varepsilon} \!\! \triangleq \!\! [f^{\varepsilon}_1 \dots f^{\varepsilon}_{K_{\varepsilon}}] \! \in \! \mathbb{C}^{K_{\varepsilon} \times 1}$ are slack variables representing the individual EEs, $\text{EE}^{\varepsilon}_{\phi}$, divided by bandwidth $B$.
{We simplify non-convex constraints~\eqref{P1const1} and \eqref{P3const1} as follows: \newline
$\bullet$ We introduce slack variables $\bm{\Psi}^{\varepsilon} \!\triangleq\! [\Psi^{\varepsilon}_1, \dots, \Psi^{\varepsilon}_{K_{\varepsilon}}], \bm{\zeta}^{\varepsilon} \!\triangleq\! [\zeta^{\varepsilon}_1, \dots, \zeta^{\varepsilon}_{K_{\varepsilon}}], \bm{\lambda}^{\varepsilon} \!\triangleq\! [\lambda^{\varepsilon}_1, \dots, \lambda^{\varepsilon}_{K_{\varepsilon}}] \! \in\! \mathbb{C}^{K_{\varepsilon} \times 1}$. representing square roots of individual SEs, $S^{\varepsilon}_{\phi}$; individual signal-to-interference-and-noise-ratios (SINRs), and square roots of numerators of individual SINRs, respectively. We substitute $\bm{C} \!\triangleq\! \{c_{mk}\!\triangleq\sqrt{\eta_{mk}}\} \!\in\! \mathbb{C}^{M \times K_d}$, for concave variable $\sqrt{\eta_{mk}}$. \newline 
$\bullet$  We now use \eqref{dlrate}-\eqref{ulrate} and~\eqref{dlpower}-\eqref{ulpower}, respectively, to expand the individual SEs, $S^{\varepsilon}_{\phi}$, and the power consumptions, $p^{\varepsilon}_{\phi}$, in terms of the optimization variables $\bm{\{C,\Theta,f^{\varepsilon}_{\phi}, \Psi^{\varepsilon}_{\phi}, \zeta^{\varepsilon}_{\varphi}, \lambda^{\varepsilon}_{\varphi}\}}$. \newline 
$\bullet$ To linearise the right-hand side of jointly convex constraints, 
at the $n$th SCA iteration, we substitute first-order Taylor approximates, $\frac{f^{2}_{1}}{f_2} \! \geq \! 2 \frac{f^{(n)}_{1}}{f^{(n)}_2} f_1 \!-\! \frac{(f^{(n)}_1)^{2}}{(f^{(n)}_2)^{2}}f_2 \!\triangleq\! \Lambda^{(n)}\! \left(\frac{f^{2}_{1}}{f_2}\right)$, as global under-estimators~\cite{boyd2004convex}, to recast {\textbf{P2}} into a GCP as follows}
\begin{subequations}
\begin{align}
\notag {\textbf{P3}}:& \!\!\!\!\!\!\!\! \underset{\substack{\bm{C\text{, }\Theta\text{, }f^d\text{, }f^u}\\ \bm{\Psi^d\text{, }\Psi^u\text{, }\zeta^d\text{, }}\\\bm{\zeta^u\text{, }\lambda^d\text{, }\lambda^u}}}{\mbox{max}} && \!\!\!\!\!\!\! \sum_{k=1}^{K_d} w^d_k f^d_k  + \sum_{l=1}^{K_u} w^u_l f^u_l \\
& \;\; \text{ s.t. } \; && \!\!\!\!\!\!\! \tau_f\log_2(1 \!+\! \zeta^d_k) \!\geq\! S^d_{ok} \text{, } \tau_f\log_2(1\!+\!\zeta^u_l) \! \geq \! S^u_{ol},\!\!\!\!\! \label{P4const5} \\
&&& \!\!\!\!\!\!\! (\Psi^d_{k})^{2} \!\leq\! \tau_f \log_2(1\!+\! \zeta^d_k) \text{, } (\Psi^u_{l})^{2} \!\leq\! \tau_f \log_2(1\!+\!\zeta^u_l), \!\! \label{P4const2} \\
&&& \!\!\!\!\!\!\! \lambda^d_k \leq \sum_{m \in \mathcal{M}^{d}_{k}} A^d_{mk} c_{mk} \text{, } (\lambda^{u}_l)^2 \leq A^u_l \theta_l, \label{P5const3} \\
&&&\!\!\!\!\!\!\!\!\!\!\! \sum_{q=1}^{K_u} \!\! B^u_{lq} \theta_q \!\!+\!\! \sum_{i=1}^{M} \!\! \sum_{k \in \kappa_{di}} \!\!\!\! D^u_{lik} c^2_{ik} \!\!+\!\! E^u_l \theta_l \!\!+\!\! F^u_l \!\!\leq\!\! \Lambda^{(n)}\!\! \left(\!\frac{(\lambda^{u}_{l})^{2}}{\zeta^u_l}\!\!\right)\!\!, \label{P6const1} \\
&&& \!\!\!\!\!\!\! \sum_{m=1}^{M}\! \sum_{q \in \kappa_{dm}} \!\!\!\! B^d_{kmq} c^2_{mq} \!\!+\!\! \sum_{l=1}^{K_u} \!\! D^d_{kl} \theta_l \!+\! 1 \!\leq\!\! \Lambda^{(n)}\! \left(\frac{(\lambda^{d}_{k})^{2}}{\zeta^d_k}\right)\!\!, \label{P6const2} \\
&&&\!\!\!\!\!\!\!\!\! P_{\text{fix}} \!\!\!+\! N_t\rho_d N_0 \!\!\!\! \sum_{m \in \mathcal{M}^{d}_{k}} \!\!\!\! \frac{\gamma^{d}_{mk} c^{2}_{mk}}{\alpha_m}\!+\! P^{d}_{\text{tc},k} \!\leq\! \Lambda^{(n)}\! \left(\frac{(\Psi^{d}_{k})^{2}}{f^d_k}\right)\!\!,  \label{P6const3} \\
&&&\!\!\!\!\!\!\! P_{\text{fix}} \!+\! \rho_u N_0 \frac{\theta_l}{\alpha'_l} \!+\! P^{u}_{\text{tc},l}  \leq \Lambda^{(n)}  \left(\frac{(\Psi^{u}_{l})^{2}}{f^u_l}\right)\!\!, \!\!\!\! \label{P6const4} \\
&&& \!\!\!\!\!\!\! \lambda^d_k \geq 0, \Tilde{b} \!\! \sum_{k \in \kappa_{dm}} \!\!\!\! \gamma^{d}_{mk} c^2_{mk} \leq \frac{1}{N_t} \text{, } c_{mk} \geq 0, \eqref{cons2}. \!\!\! \label{P5const4}
\end{align}
\end{subequations}
The residue after the $n$th iteration, $\bm{r}_{\text{SCA}}^{(n)}$, has a magnitude
\[\|\bm{r}_{\text{SCA}}^{(n)}\| = \sqrt{\|\bm{C}^{(n+1)} - \bm{C}^{(n)}\|^{2}_{F} + \|\bm{\Theta}^{(n+1)} - \bm{\Theta}^{(n)}\|^{2}}.\]  
{We provide an iterative method to solve {\textbf{P3}} in Algorithm~\ref{algo1}.}
\vspace{-0.2cm}
\begin{algorithm} 
	\footnotesize
\DontPrintSemicolon 
\KwIn{i) Initialize power control coefficients $\{\bm{C,\Theta}\}^{(1)}$ by allocating equal power to all downlink UEs being served and full power to all uplink UEs.  Set $n = 1$.\\
ii) Initialize $\{\bm{f^d,f^u,\Psi^d,\Psi^u,\zeta^d,\zeta^u,\lambda^d,\lambda^u}\}^{(1)}$ by replacing \eqref{P5const3}, \eqref{P6const1}-\eqref{P6const2}, \eqref{P4const2} and \eqref{P6const3}-\eqref{P6const4} by equality.}
\KwOut{Globally optimal power control coefficients $\{\bm{C,\Theta}\}^{*}$}
\While{$\|\bm{r}_{\text{SCA}}^{(n)}\| \leq \epsilon_{\text{SCA}}$ {\text{(convergence threshold)}}}{
Solve {\textbf{P3}} for the $n$th SCA iteration to obtain optimal variables, $\{\bm{f^d, f^u, \Psi^d, \Psi^u, \zeta^d, \zeta^u, \lambda^d, \lambda^u, C, \Theta}\}^{*,(n)}$. \\
Assign the SCA iterates for the $(n+1)$th iteration, $\{\bm{f^d, f^u, \Psi^d, \Psi^u, \zeta^d, \zeta^u, \lambda^d, \lambda^u, C, \Theta}\}^{(n+1)} = \{\bm{f^d, f^u, \Psi^d, \Psi^u, \zeta^d, \zeta^u, \lambda^d, \lambda^u, C, \Theta}\}^{*,(n)}$. 
}
\caption{WSEE maximization algorithm}\label{algo1}
\end{algorithm}
\vspace{-0.4cm}
\section{Simulation results} \label{simresults}
\vspace{-0.2cm}
We now numerically investigate the performance of a FD CF mMIMO system with limited capacity fronthaul links using the proposed 
optimization approach. We assume a realistic system model wherein $M$ APs, $K_d$ downlink UEs and $K_u$ uplink UEs are scattered in a square of size $D$ km $\times$ $D$ km. To avoid the boundary effects~\cite{CFvsSmallCells}, we wrap the APs and UEs around the edges~\cite{FDCellFree}.  
The large-scale fading coefficients {for the channels between the APs and the uplink and downlink UEs, denoted as $\beta^{d}_{mk}$ and $\beta^{u}_{ml}$ respectively,} are modeled as~\cite{CFvsSmallCells}
\begin{align}
\beta^{d}_{mk} &= 10^{\frac{\text{PL}^{d}_{mk}}{10}} 10^{\frac{\sigma_{\text{sd}} z^{d}_{mk}}{10}} \text{ and } \beta^{u}_{ml} &= 10^{\frac{\text{PL}^{u}_{ml}}{10}} 10^{\frac{\sigma_{\text{sd}} z^{u}_{ml}}{10}}. \label{pathloss}
\end{align}
Here {$10^{\frac{\sigma_{\text{sd}} z^{d}_{mk}}{10}}$ and $10^{\frac{\sigma_{\text{sd}} z^{u}_{ml}}{10}}$ are log-normal shadowing factors} having a standard deviation $\sigma_{\text{sd}}$ (in dB) and {$z^{d}_{mk}$ and $z^{u}_{ml}$ follow} a two-components correlated model~\cite{CFvsSmallCells}. The path {losses, $\text{PL}^{d}_{mk}$ and $\text{PL}^{u}_{ml}$} (in dB) follows a three-slope model~\cite{CFvsSmallCells, FDCellFree}. 
We, similar to~\cite{FDCellFree}, model the large-scale fading coefficients for the inter-AP RI channels, i.e., $\beta_{\text{RI}, mi}, \, \forall i \neq m$, as in \eqref{pathloss}, and assume that the large-scale fading for the intra-AP RI channels, which experience no shadowing, are modeled as $\beta_{\text{RI},mm} = 10^{\frac{\text{PL}_{\text{RI}} \text{(dB)}}{10}}$. The inter-UE large scale fading coefficients, $\Tilde{\beta}_{kl}$, are also modeled similar to \eqref{pathloss}. We consider, for brevity, the same number of quantization bits $\nu$, and the same capacity $C_{\text{fh}}$, on all fronthaul links. We, henceforth, denote the transmit powers on the downlink and uplink as $p_{d}$ $(= \rho_{d} N_0)$ and $p_{u}$ $(= \rho_{u} N_0)$ respectively, and the pilot transmit power as $p_t (= \rho_t N_0)$.  We fix the system model and power consumption model parameters, similar to~\cite{CFvsSmallCells,FDCellFree,CellFreeEEUQ}, as in Table \ref{sysparams}.
\begin{table}[hbt!]
\small
\centering
 \begin{tabular}{|c|c|} 
 \hline
 Parameter & Value \\ 
 \hline 
  $D$, $\tau_c$, $T_c$ & 1 km, 200, 1 ms \\
  \hline
  $\sigma_{\text{sd}}$, B & 2 dB, 20 Mhz \\ 
 \hline
  Fronthaul parameters $\nu$, $C_{\text{fh}}$ & $2$, $100$ Mbps \\
  \hline 
  $\gamma_{\text{RI}}$, $\text{PL}_{\text{RI}}$ (in dB) & $-20$, $-81.1846$ \\
  \hline 
  $P_{\text{ft}}, P_{0,m}, P_{\text{tc},m} = P^{d}_{\text{tc},k} = P^{u}_{\text{tc},l}, p_t$ (in W) & 10, 0.825, 0.2, 0.2 \\
  \hline
  $N_0, \alpha_m = \alpha'_l$ & -121.4 dB, 0.4 \\
 \hline
\end{tabular}
\caption{FD CF mMIMO system and power consumption parameters}
\label{sysparams}
\end{table}
\begin{figure*}[ht]
    \centering
    \begin{subfigure}{.3\textwidth}
      \centering
          \includegraphics[height = \linewidth, width=\linewidth]{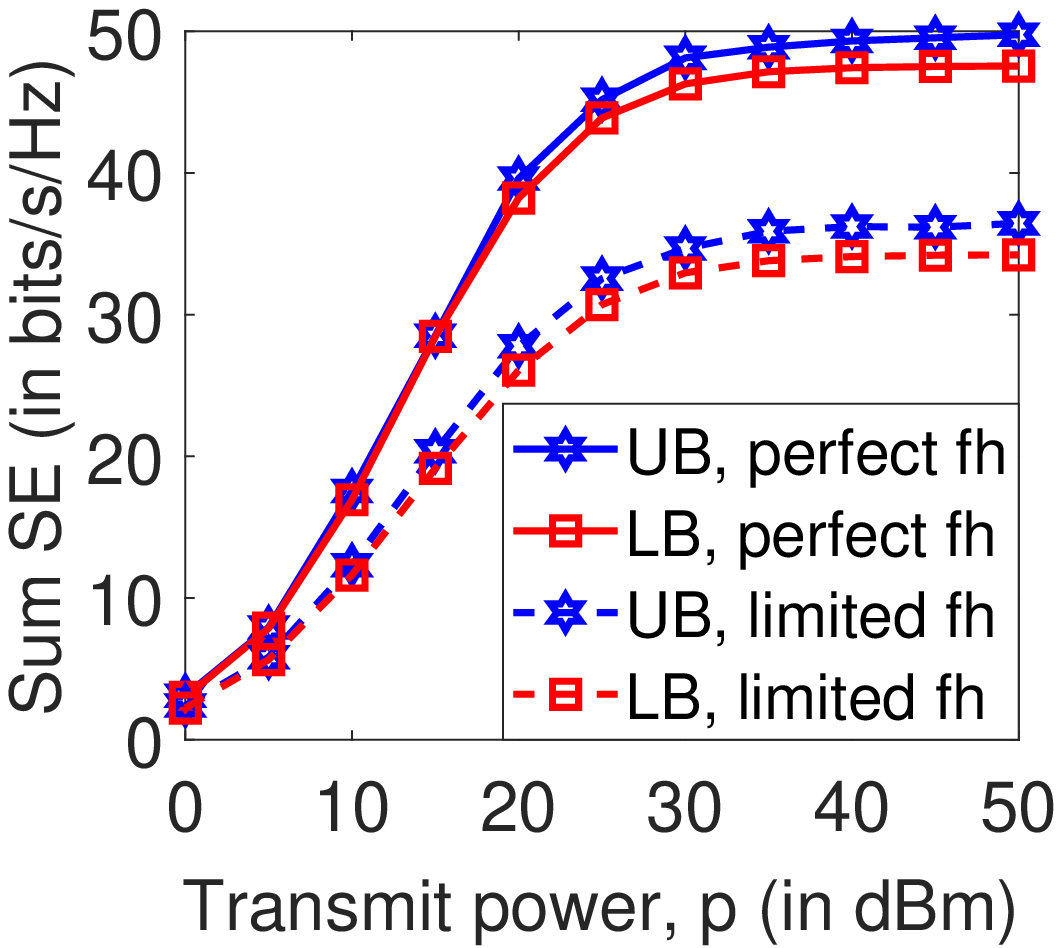}
          \caption{} 
          \label{fig:2a}
    \end{subfigure}
    \begin{subfigure}{.3\textwidth}
		\centering
		\includegraphics[height = \linewidth, width=\linewidth]{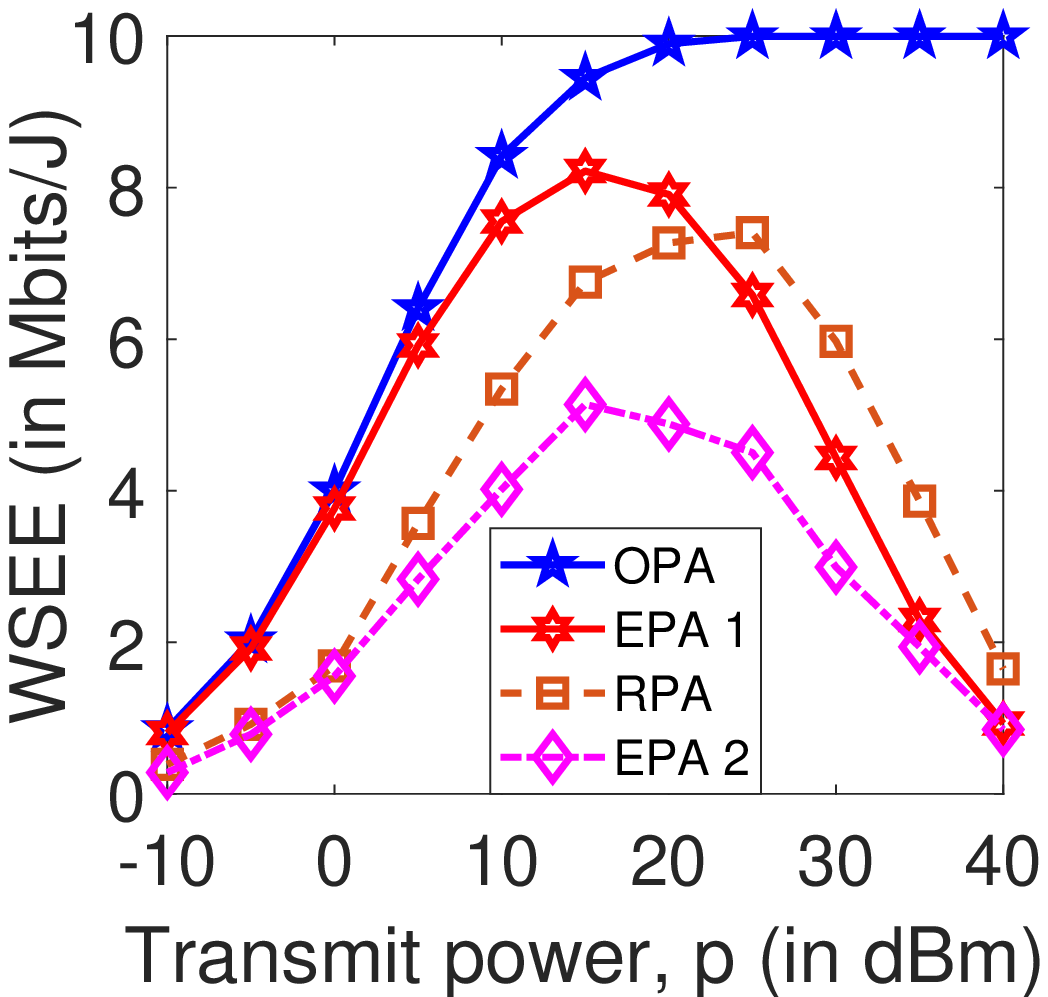}  
		\caption{} 
		\label{fig:6a}
	\end{subfigure}
	 \begin{subfigure}{.3\textwidth}
		\centering
		\includegraphics[height = \linewidth, width = \linewidth]{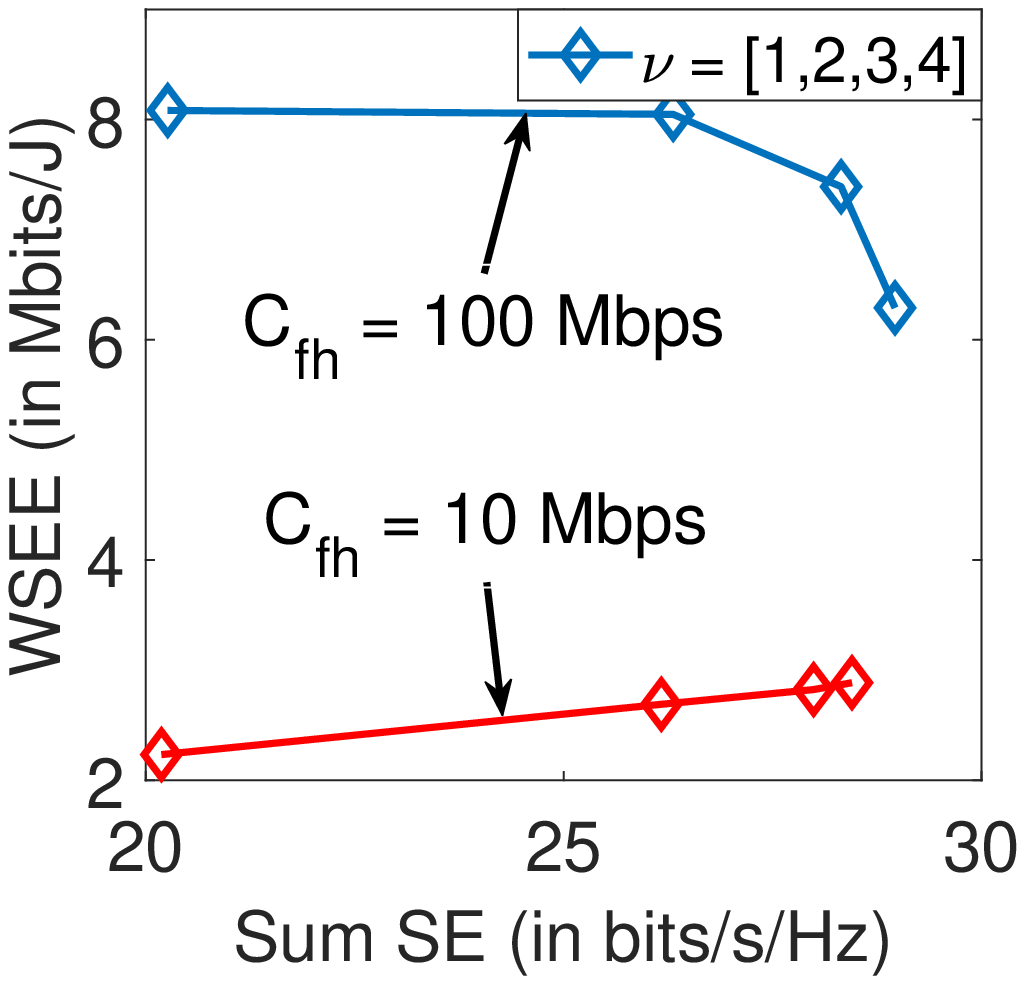}
		\caption{}
		\label{fig:7b}
	\end{subfigure}
\vspace{-0.4cm}
\caption{a) Sum SE (analytical LB and ergodic UB) vs transmit power with $M = 32, N_t = N_r = 2, K_d = K_u = 10$;  (b) WSEE vs maximum transmit power and (c) WSEE vs sum SE by varying $\nu = 1 \text{ to } 4$, with $M = 32, K_d = K_u = 10, N_t = N_r = 2$ and $S_{ok} = S_{ol} = 0.1$ bits/s/Hz; 
\vspace{-0.6cm}}
\label{fig:2}
\end{figure*}

\vspace{-0.2cm}
\textbf{Validation of achievable SE expressions:} We consider an FD CF mMIMO system with $M \!\!=\!\! 32$ APs, each having $N_t\!\! = \!\!N_r \!\! = \!\!2$ transmit and receive antennas, and $K_d \!\!=\!\! K_u \!\! = \!\! 10$ uplink and downlink UEs and consider equal transmit power for uplink and downlink data transmission, i.e., $p_d = p_u = p$. We verify in Fig.~\ref{fig:2a} the tightness of the derived SE lower bound in \eqref{dlrate}-\eqref{ulrate}, labeled as LB, by comparing it with the numerically obtained ergodic SE 
labeled as upper-bound (UB) as it requires instantaneous CSI. The large-scale fading coefficients are set to unity, i.e., $\beta^{d}_{mk}\!\! = \!\! \beta^{u}_{ml} \!\!=\!\! 1, \forall k\!\! \in\!\! \kappa_{dm}, l \!\!\in \!\! \kappa_{um}$. We, similar to \cite{CFvsSmallCells}, allocate equal power to all downlink UEs  and full power to all uplink UEs, i.e., $\eta_{mk}\!\! =\!\! \left({b}N_t\left(\sum_{k \in \kappa_{dm}} \!\! \gamma^{d}_{mk}\right)\right)^{-1}, \forall  k \!\! \in \!\! \kappa_{dm}$ and $\theta_l \!=\! 1$. We consider two cases: i) perfect high-capacity fronthaul  with  $\Tilde{a} \! =\! \Tilde{b} \!=\! 1$, and ii) limited fronthaul links, with $\nu \!=\! 2$ quantization bits and capacity $C_{\text{fh}} \!=\! 10$ Mbps. We see that  the derived lower bound is tight for both {perfect and limited} fronthauls. With limited fronthaul, the sum SE significantly drops, as quantization distorts the signal and limits the number of UEs that each AP can serve (see~\eqref{maxUEs}).

\textbf{WSEE maximization - optimal system parameters:} 
We now study  WSEE variation with system parameters. We consider $M \!\!=\!\! 32$ APs each having $N_t \!\!=\!\! N_r \!\!=\!\! 2$ 
antennas, $K_u \!\!=\!\! K_d \!\!=\!\! 10$ downlink and uplink UEs and QoS constraints $S_{ok} \!\!=\!\! S_{ol} \!\!=\!\! 0.1$ bits/s/Hz, unless mentioned otherwise.

We plot in Fig. \ref{fig:6a} the WSEE versus transmit power $p$, which is taken to be equal for downlink and uplink, i.e., $p_d \!=\! p_u \!=\! p$. We consider optimal power allocation (OPA) in Algorithm~\ref{algo1}, labeled as ``OPA", and compare {it} with three sub-optimal power allocation schemes: i) equal power allocation of type 1, labeled as ``EPA 1", where $\eta_{mk} \!\!=\!\! \left({b}N_t\left(\sum_{k \in \kappa_{dm}}\!\!\!\! \gamma^{d}_{mk}\right)\right)^{-1}, \forall k \!\! \in \!\! \kappa_{dm}$ and $\theta_l \!=\! 1$~\cite{CellFreeEE, CellFreeEEUQ}, ii) equal power allocation of type 2, labeled as ``EPA 2", where $\eta_{mk}\!\! =\!\! \left({b} N_t K_{dm} \gamma^{d}_{mk}\right)^{-1}, \forall k \!\!\in\!\! \kappa_{dm}$ and $\theta_l \!=\! 1$~\cite{CellFreeEE}, and iii) random power allocation, labeled as ``RPA", where power control coefficients are chosen randomly from a uniform distribution between $0$ and the ``EPA 1" value. {We note that existing literature has not yet optimized the WSEE metric for CF mMIMO systems, hence we can only compare with above sub-optimal schemes. We note that our proposed approach far outperforms the baseline schemes.}

{
We next quantify in Fig. \ref{fig:7b} the joint variation of WSEE and the sum SE, with the number of quantization bits $\nu$, in the fronthaul links. We note that the WSEE is calculated using Algorithm~\ref{algo1}. We consider transmit power $p_d = p_u = p = 30$ dBm 
and take two different cases: i) high fronthaul capacity, $C_{\text{fh}} = 100$ Mbps, which is sufficiently high to support all the UEs, and ii) limited fronthaul capacity, $C_{\text{fh}} = 10$ Mbps, which limits the number of UEs a single AP can serve. We observe that for $C_{fh} = 100$ Mbps, the WSEE falls with increase in $\nu$, even though the corresponding sum SE increases. For $C_{fh} = 10$ Mbps, the sum SE and the WSEE simultaneously increase with increase in $\nu$. To explain this behaviour, we note that increasing $\nu$ improves the sum SE for both the cases due to a reduction in the quantization attenuation and distortion. For $C_{fh} = 100$ Mbps, the APs serve all the UEs, i.e., $K_{dm} = K_d$ and $K_{um} = K_u$, so increasing $\nu$ linearly increases the fronthaul data rate, $R_{fh}$ (see~\eqref{fhrate}). This, as seen from~\eqref{fixpower}, increases the traffic-dependent fronthaul power consumption. Using lower number (1-2) of quantization bits is therefore more energy-efficient, as it provides sufficiently good SE with a low energy consumption. However, for $C_{fh} = 10$ Mbps, $K_{um}$ and $K_{dm}$ have an upper limit, given by~\eqref{maxUEs}, which is inversely related to $\nu$. The product, $\nu (K_{um} + K_{dm})$, remains nearly constant for all values of $\nu$. Thus, $R_{fh}$ (see~\eqref{fhrate}) doesn't increase with increase in $\nu$ and remains close to the capacity, $C_{fh}$. The traffic-dependent fronthaul power consumption, given in~\eqref{fixpower}, hence, remains close to $P_{\text{ft}}$. A higher number (3-4) of quantization bits therefore provides a higher sum SE and hence, also maximizes the WSEE.
}

\vspace{-0.1cm}
\section{Conclusion}
\vspace{-0.1cm}
We derived the achievable SE expressions for a FD CF mMIMO wireless system with   fronthaul quantization. We optimized the non-convex WSEE using SCA framework which in each iteration solves a GCP. 
We numerically investigated the WSEE dependence  on various system parameters which gives important insights to design energy-efficient systems. 
%
\appendices 
\section{} \label{SINRterms}
{
We now derive the achievable SE  for the $k$th downlink UE in~\eqref{dlrate} and the $l$th uplink UE in~\eqref{ulrate}. We use $\varepsilon \!\triangleq\! \{d,u\}$ to denote downlink and uplink, respectively; $\xi \!\triangleq\! \{t,r\}$ to denote transmit and receive, respectively; $\phi \!\triangleq\! \{k,l\}$ to denote the $k$th downlink UE and the $l$th uplink UE, respectively; and $\upsilon^{\varepsilon}_{m\phi} \!\triangleq\! \{\eta_{mk} \text{ for } \phi\!=\!k, \theta_{l} \!\text{ for }\! \phi\!=\!l\}$. From Section~\ref{ul_ch_est}, we know that $\bm{g}^{\varepsilon}_{m\phi} = \hat{\bm{g}}^{\varepsilon}_{m\phi} + \bm{e}^{\varepsilon}_{m\phi}$, where $\hat{\bm{g}}^{\varepsilon}_{m\phi}$ and $\bm{e}^{\varepsilon}_{m\phi}$ are independent and $\mathbb{E}\{\|\hat{\bm{g}}^{\varepsilon}_{m\phi}\|^{2}\} = N_{\xi}\gamma^{\varepsilon}_{m\phi}$.} 
{We express the desired signal for the $k$th downlink UE and the $l$th uplink UE as 
\begin{small}
\begin{align}
    \notag \mathbb{E}\{|\text{DS}^{\varepsilon}_{\phi}|^{2}\} &= \! \Tilde{a}^{2} \rho_{\varepsilon} \mathbb{E}\{\Big| \!\!\sum_{m \in \mathcal{M}^{\varepsilon}_{\phi}} \!\! \sqrt{\upsilon^{\varepsilon}_{m\phi}} \mathbb{E}\{(\hat{\bm{g}}^{\varepsilon}_{m\phi})^{T} (\hat{\bm{g}}^{\varepsilon}_{m\phi})^{*}\} s^{\varepsilon}_{\phi}\Big|^{2}\} \nonumber \\
    &= \Tilde{a}^{2} N^2_{\xi} \rho_{\varepsilon} (\sum_{m \in \mathcal{M}^{\varepsilon}_{\xi}} \!\! \sqrt{\upsilon^{\varepsilon}_{m\phi}} \gamma^{\varepsilon}_{m\phi})^{2}. \label{dl_DS} 
\end{align}
\end{small}
We now calculate the beamforming uncertainty for the $k$th downlink UE and the $l$th uplink UE as $\mathbb{E}\{|\text{BU}^{\varepsilon}_{\phi}|^{2}\}$
\begin{small}
\begin{align}
\notag  \!&=\! \Tilde{a}^{2} \rho_{\varepsilon}\!\!\!\! \sum_{m \in \mathcal{M}^{\varepsilon}_{\phi}} \!\!\!\! \mathbb{E}\{|\sqrt{\upsilon^{\varepsilon}_{m\phi}}((\bm{g}^{\varepsilon}_{m\phi})^{T} (\hat{\bm{g}}^{\varepsilon}_{m\phi})^{*} - \mathbb{E}\{(\bm{g}^{\varepsilon}_{m\phi})^{T} (\hat{\bm{g}}^{\varepsilon}_{m\phi})^{*})\})|^2\} \\
& \!\stackrel{(a)}{=} \! \Tilde{a}^{2} \rho_{\varepsilon} \!\!\!\!\! \sum_{m \in \mathcal{M}^{\varepsilon}_{\phi}} \!\!\!\!\! \upsilon^{\varepsilon}_{m\phi} \! (N_{\xi}(N_{\xi} \!+\!1) (\gamma^{\varepsilon}_{m\phi})^2 \!\! + \!\! N_{\xi} \gamma^{\varepsilon}_{m\phi}(\beta^{\varepsilon}_{m\phi} \!\! - \!\! \gamma^{\varepsilon}_{m\phi}) \!\! - \!\! (N_{\xi}\gamma^{\varepsilon}_{m\phi})^2) \nonumber \\
&= \Tilde{a}^{2} N_{\xi} \rho_{\varepsilon} \sum_{m \in \mathcal{M}^{\varepsilon}_{\phi}} \upsilon^{\varepsilon}_{m\phi} \beta^{\varepsilon}_{m\phi} \gamma^{\varepsilon}_{m\phi}. \label{dl_BU} 
\end{align}
\end{small}
Equality $(a)$ is obtained by i) using $\mathbb{E}\{|s^{\varepsilon}_{\phi}|^2\} = 1$; ii) substituting $\bm{g}^{\varepsilon}_{m\phi} = \hat{\bm{g}}^{\varepsilon}_{m\phi} + \bm{e}^{\varepsilon}_{m\phi}$;  ii) using the fact that $\hat{\bm{g}}^{\varepsilon}_{m\phi}$ and $\bm{e}^\varepsilon_{m\phi}$ are zero-mean and uncorrelated; v) using the results $\mathbb{E}\{\|\hat{\bm{g}}^{\varepsilon}_{m\phi}\|^{4}\} = N_{\xi}(N_{\xi}+1) (\gamma^{\varepsilon}_{m\phi})^{2}$~\cite{RandomMatrix} and $\mathbb{E}\{\|\bm{e}^{\varepsilon}_{m\phi}\|^{2}\} = (\beta^{\varepsilon}_{m\phi} - \gamma^{\varepsilon}_{m\phi})$.} 
{The powers of the remaining interference terms (MUI, RI, UDI), noise terms (N) and the quantization distortion terms (TQD), can be obtained using similar calculations as above, which are omitted for the sake of space. 
}
\bibliographystyle{IEEEtran}
\bibliography{IEEEabrv,paper_template_ref}
\end{document}